\documentclass[12pt, oneside, reqno]{amsart}

\pdfoutput=1
\usepackage{amsfonts}
\usepackage{amssymb}
\usepackage{amsmath}
\usepackage{amsrefs}
\usepackage{tikz-cd}
\usepackage{centernot}
\usepackage{enumerate}
\usepackage{fullpage}
\usepackage{multicol}
\usepackage{amsthm,amsmath,amssymb,tikz,tikz-cd,amsmath,mathrsfs,mathtools,multicol,dirtytalk,tabularx,xy,lipsum,url,enumitem,cmll,tkz-euclide}
\usepackage{float}

\newtheorem{thm}{Theorem}[section]
\newtheorem{lem}[thm]{Lemma}
\newtheorem{cor}[thm]{Corollary}

\theoremstyle{definition}
\newtheorem{defn}[thm]{Definition}

\theoremstyle{remark}

\newcommand{\mc}{\mathcal}
\newcommand{\<}{\langle}
\renewcommand{\>}{\rangle}
\newcommand{\ol}{\overline}
\newcommand{\ts}{\textsc}
\newcommand{\C}{\mathbb{C}}
\newcommand{\R}{\mathbb{R}}

\author{John Harding}
\address{New Mexico State University, Las Cruces NM 88003}
\email{jharding@nmsu.edu}

\author{Remi Salinas Schmeis}
\address{New Mexico State University, Las Cruces NM 88003}
\email{remis@nmsu.edu}

\thanks{The first listed author was partially supported by US Army grant W911NF-21-1-0247 and NSF grant DMS-2231414. The second listed author was partially supported by NSF grant DMS-2231414. }

\begin{document}

\title{Remarks on orthogonality spaces}

\begin{abstract}
We provide two results. The first gives a finite graph constructed from consideration of mutually unbiased bases that occurs as a subgraph of the orthogonality space of $\mathbb{C}^3$ but not of that of $\mathbb{R}^3$. The second is a companion result to the result of Tau and Tserunyan \cite{Tau} that every countable graph occurs as an induced subgraph of the orthogonality space of a Hilbert space. We show that every finite graph occurs as an induced subgraph of the orthogonality space of a finite orthomodular lattice and that every graph occurs as an induced subgraph of the orthogonality space of some atomic orthomodular lattice. 
  
\par
\vspace{.2cm}
\par
\vspace{.2cm}
\end{abstract}

\maketitle
\maketitle

\section{Introduction}

Orthomodular posets (\ts{omp}s) \cite{ptak} are a certain type of bounded poset with complementarity involution arising as models of the propositions of a quantum system. Orthomodular lattices (\ts{oml}s) \cite{Kalmbach} are \ts{omp}s that are lattices. The motivating example is the projection lattice $\mc{P}(\mc{H})$ of a Hilbert space. Each atomic orthomodular poset $P$ yields a symmetric loopless graph called its  orthogonality space $(X,\perp)$ where $X$ is the set of atoms of $P$ and $\perp$ is the orthogonality relation between atoms. For a Hilbert space $\mc{H}$, the orthogonality space associated to $\mc{P}(\mc{H})$ is given by the 1-dimensional subspaces of $\mc{H}$ with orthogonality being usual orthogonality of subspaces. This orthogonality space represents the pure states of a quantum system represented by $\mc{H}$. 

A graph $G$ is a subgraph of a graph $H$ if the vertex and edge sets of $G$ are subsets of those of $H$ and $G$ is an induced subgraph of $H$ if additionally the edge set of $G$ is the restriction of the edge set of $H$ to the vertices of $G$. By a graph embedding $f:G\to H$ we mean a one-one mapping where $x$ adjacent to $y$ implies $f(x)$ is adjacent to $f(y)$, and this embedding is full if additionally $f(x)$ adjacent to $f(y)$ implies that $x$ is adjacent to $y$. So an embedding is a graph isomorphism with a subgraph of $H$ and a full embedding is a graph isomorphism with an induced subgraph of $H$. Our purpose here is to provide two results related to graph embedding and orthogonality spaces of atomic \ts{oml}s.


The first result gives a finite graph that occurs as a subgraph of the orthogonality space of the \ts{oml} $\mc{P}(\mathbb{C}^3)$ but does not occur in the orthogonality space of $\mc{P}(\mathbb{R}^3)$. That there is such a finite graph is not a surprise, the orthogonality space of $\mc{P}(\C^3)$ determines the order structure of $\mc{P}(\C^3)$ in a simple way, and well-known coordinatization theorems of projective geometry \cite{Dilworth} allow one to reconstruct the field $\C$ from the order structure of $\mc{P}(\C^3)$ via the coordinatization theorem. Thus, one can find a configuration showing that $-1$ has a square root in $\mc{P}(\C^3)$, but of course not in $\mc{P}(\R^3)$. However, construction of such a configuration and its expression in the orthogonality space setting will be somewhat complicated. We provide a relatively simple configuration of atoms, expressed in the language of Greechie diagrams, that characterizes when one ray $a$ of $\R^3$ or $\C^3$ is unbiased with respect to three rays $b,c,d$ that arise from an orthonormal basis, that is, when the angles between $a$ and $b,c,d$ are equal. Using the fact that $\mathbb{C}^3$ has a pair of mutually unbiased bases while $\R^3$ does not, we obtain a configuration that exists in the orthogonality space of $\mc{P}(\C^3)$ but not in that of $\mc{P}(\R^3)$. 

The second result deals with full graph embeddings of graphs into orthogonality graphs of atomic \ts{oml}s. Tau and Tserunyan \cite{Tau} proved that every at most countable graph can be fully embedded into the orthogonality space of a separable Hilbert space, and their construction shows that a finite graph with maximum clique size $n$ can be embedded into the orthogonality graph of $\mc{P}(\R^n)$ or $\mc{P}(\C^n)$. They provide an example of an uncountable graph that cannot be fully embedded into the orthogonality graph of $\mc{P}(\mc{H})$ for any Hilbert space $\mc{H}$. We provide a complementary result, that any finite graph can be fully embedded into the orthogonality graph of a finite \ts{oml}. The result of Tau and Tserunyan fully embeds a finite graph into the orthogonality space of a finite height \ts{oml}, but of course not into a finite \ts{oml}. Either our result, or that of Tau and Tserunyan, can be used to show that every graph can be fully embedded into the orthogonality graph of some atomic \ts{oml}. 


\section{The first result}

In this section, we provide a finite graph that can be embedded into the orthogonality graph of the projection lattice $\mc{P}(\mathbb{C}^3)$ but not into that of $\mc{P}(\mathbb{R}^3)$. Orthogonality graphs can be represented in the usual graph-theoretic way using vertices and edges, or they can be represented by ``Greechie diagrams'' \cite{Kalmbach} which do not directly provide the edges, but instead describe the maximal cliques of the graph, which of course is equivalent to describing the edges. In the following, by a Greechie diagram we will mean a collection of vertices and a collection of sets of those vertices that arises as the vertex set and the set of maximal cliques of some graph. 
We find Greechie diagrams more convenient to use in our context than graphs and they are the customary tool when working with orthogonality spaces. For example, in the figure at left below is a Greechie diagram showing a graph with two maximal cliques, which we call blocks, of three elements each, $\{a,b,c\}$ and $\{c,d,e\}$. This same graph is shown in the usual graph-theoretic way in the diagram at right. We note that maximal cliques, or blocks, in our Greechie diagrams are given by points connected with a straight line or arc with no sharp bends. 
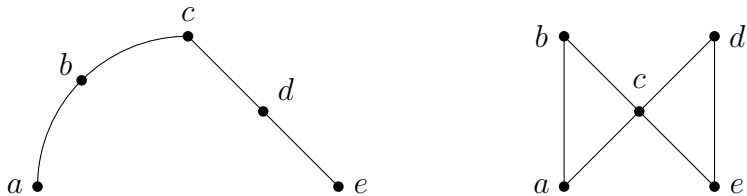
\vspace{2ex}

\begin{center}
\begin{tikzpicture}
\fill (90:2) circle(.07);
\fill (135:2) circle(.07);
\fill (180:2) circle(.07);
\draw (90:2) arc(90:180:2);
\fill (1,1) circle(.07);
\fill (2,0) circle(.07);
\draw (0,2)--(2,0);
\node at (180:2.3) {$a$};
\node at (135:2.3) {$b$};
\node at (90:2.3) {$c$};
\node at (1.3,1.3) {$d$};
\node at (2.3,0) {$e$};
\fill (6,1) circle(.07);
\fill (5,0) circle(.07);
\fill (5,2) circle(.07);
\fill (7,0) circle(.07);
\fill (7,2) circle(.07);
\draw (7,2)--(5,0)--(5,2)--(7,0)--(7,2);
\node at (6,1.4) {$c$};
\node at (4.7,0) {$a$};
\node at (4.7,2) {$b$};
\node at (7.3,2) {$d$};
\node at (7.3,0) {$e$};
\end{tikzpicture}    
\end{center}
\vspace{2ex}

Since we are dealing both with $\mathbb{R}^3$ and $\mathbb{C}^3$, it is efficient to use notation that allows us to treat both situation at the same time, so we talk of the conjugate $\ol{z}$ of a real or complex number $z$. When $z$ is real, $\ol{z}$ is simply $z$. 

\begin{defn}
If $u=(u_1,u_2,u_3)$ and $v=(v_1,v_2,v_3)$ are vectors in $\mathbb{R}^3$ or $\mathbb{C}^3$ we define their ``cross product''   
\[
u\times v\,\,=\,\, (\ol{u_2v_3-u_3v_2},\ol{u_3v_1-u_1v_3},\ol{u_1v_2-u_2v_1}).
\]
\end{defn}

For $\mathbb{R}^3$, this is the usual cross product. For $\mathbb{C}^3$, it is known that there is no operation that satisfies all common properties of the usual cross product, however we we do have the following property, which is what we will need. 

\begin{lem} \label{lem: cross product}
Let $u,v$ be linearly independent vectors in $\mathbb{R}^3$ or $\mathbb{C}^3$. Then for $w=u\times v$ we have that $w$ is non-zero vector that is orthogonal to both $u$ and $v$ and $\langle w\rangle$ is the unique 1-dimensional subspace that is orthogonal to both $\langle u\rangle$ and $\langle v\rangle$.     
\end{lem}

\begin{proof}
That $w$ is non-zero and orthogonal to both $u$ and $v$ are routine calculations. For the statement about uniqueness, if $\langle w\rangle$ and $\langle w'\rangle$ were distinct 1-dimensional subspaces orthogonal to both $\langle u\rangle$ and $\langle v\rangle$, then their sum would be a 2-dimensional subspace orthogonal to both $\langle u\rangle$ and $\langle v\rangle$, hence their sum would be equal to $\langle u\rangle^\perp$ and to $\langle v\rangle^\perp$, giving $\langle u\rangle = \langle v\rangle$.     
\end{proof}

We describe some basic properties of the Greechie diagrams of $\mc{P}(\mathbb{R}^3)$ and $\mc{P}(\mathbb{C}^3)$ that are well known \cite{Kalmbach}. First, the diagram of $\mc{P}(\mathbb{R}^3)$ can be embedded into that of $\mc{P}(\mathbb{C}^3)$. For either diagram, all of its blocks consist of 3 pairwise orthogonal 1-dimensional subspaces, in effect, an orthonormal basis up to phase factors. Two blocks are either disjoint or share exactly one atom. An $n$-loop in a Greechie diagram is a sequence $B_1,\ldots,B_n$ of distinct blocks so that any $B_i$ and $B_{i+1}$ share an atom where addition is taken modulo $n$. It is well known, and the proof of uniqueness in Lemma~\ref{lem: cross product} shows, that the Greechie diagrams of $\mc{P}(\mathbb{R}^3)$ and $\mc{P}(\mathbb{C}^3)$ have no loops of order 4 or less.  
Further properties of these Greechie diagrams come from Lemma~\ref{lem: cross product}. Suppose $\langle u\rangle$ and $\langle v\rangle$ are distinct 1-dimensional subspaces, and $w=u\times v$. If $\langle u\rangle$ and $\langle v\rangle$ are orthogonal, then $\langle w\rangle$ is the third element comprising a block with  $\langle u\rangle$ and $\langle v\rangle$, and if $\langle u\rangle$ and $\langle v\rangle$ are not orthogonal, then $\langle w\rangle$ is the unique element in the intersection of two blocks, one containing $\langle u\rangle$ and $\langle w\rangle$ and the other containing $\langle v\rangle$ and $\langle w\rangle$.
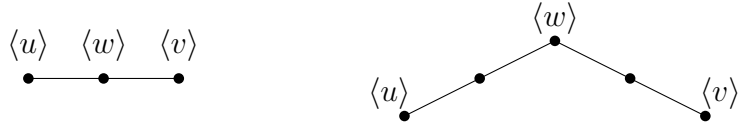
\vspace{2ex}

\begin{center}
\begin{tikzpicture}
\fill (-1,0) circle(0.07); 
\fill (0,0) circle(0.07); 
\fill (1,0) circle(0.07); 
\draw (-1,0)--(1,0);
\node at (-1,.45) {$\langle u\rangle$};
\node at (0,.45) {$\langle w\rangle$};
\node at (1,.45) {$\langle v\rangle$};
\fill (4,-.5) circle(0.07); 
\fill (5,0) circle(0.07); 
\fill (6,.5) circle(0.07); 
\fill (7,0) circle(0.07); 
\fill (8,-.5) circle(0.07); 
\draw (4,-.5)--(6,.5)--(8,-.5);
\node at (3.8,-.2) {$\langle u\rangle$};
\node at (8.2,-.2) {$\langle v\rangle$};
\node at (6,.8) {$\langle w\rangle$};
\end{tikzpicture}    
\end{center}
\vspace{2ex}

We are dealing with Greechie diagrams to represent graphs, and say that one Greechie diagram is a sub-diagram of another if the graph associated with the first Greechie diagram is a subgraph of the graph associated with the other. This means that the vertices of the first Greechie diagram are contained in the second, and whenever two vertices are part of a block in the first Greechie diagram, they are part of a block in the second. 

In the following, for a vector such as $v=(x,y,z)$, to improve readability we write $\<x,y,z\>$ for the one-dimensional subspace spanned by $v$ in place of the correct $\<(x,y,z)\>$.

\begin{center}
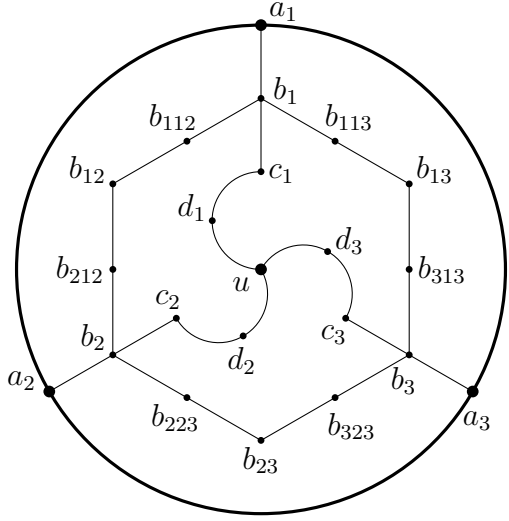
\begin{figure}[h]
\centering
\begin{tikzpicture}[scale=.65]
\draw[very thick] (0,0) circle(5);  
\fill (0:0) circle(.12);
\fill (90:5) circle(.12);
\fill (210:5) circle(.12);
\fill (-30:5) circle(.12);
\fill (90:3.5) circle(.07);
\fill (90:2) circle(.07);
\fill (210:3.5) circle(.07);
\fill (210:2) circle(.07);
\fill (-30:3.5) circle(.07);
\fill (-30:2) circle(.07);
\fill (135:1.41) circle(.07);
\fill (15:1.41) circle(.07);
\fill (255:1.41) circle(.07);
\draw (90:2) arc(90:270:1);
\draw (90:2)--(90:5);
\draw (-30:2) arc(-30:150:1);
\draw (-30:2)--(-30:5);
\draw (210:2) arc(210:390:1);
\draw (210:2)--(210:5);
\node at (-.4,-.3) {$u$};
\node at (85:5.3) {$a_1$};
\node at (82:3.7) {$b_1$};
\node at (78:2) {$c_1$};
\node at (205:5.4) {$a_2$};
\node at (202:3.7) {$b_2$};
\node at (198:2) {$c_2$};
\node at (-35:5.4) {$a_3$};
\node at (-38:3.7) {$b_3$};
\node at (-42:2) {$c_3$};
\fill (150:3.5) circle(.07);
\fill (180:3.03) circle(.07);
\fill (120:3.03) circle(.07);
\fill (60:3.03) circle(.07);
\fill (30:3.5) circle(.07);
\fill (0:3.03) circle(.07);
\fill (-60:3.03) circle(.07);
\fill (270:3.5) circle(.07);
\fill (240:3.03) circle(.07);
\draw (90:3.5)--(150:3.5)--(210:3.5)--(270:3.5)--(-30:3.5)--(30:3.5)--(90:3.5);
\node at (150:4.1) {$b_{12}$};
\node at (30:4.1) {$b_{13}$};
\node at (270:3.9) {$b_{23}$};
\node at (18:1.9) {$d_3$};
\node at (138:1.9) {$d_1$};
\node at (258:1.9) {$d_2$};
\node at (120:3.6) {$b_{112}$};
\node at (60:3.6) {$b_{113}$};
\node at (180:3.7) {$b_{212}$};
\node at (0:3.7) {$b_{313}$};
\node at (240:3.5) {$b_{223}$};
\node at (300:3.7) {$b_{323}$};
\node at (270:5.5) {$ $};
\end{tikzpicture}   
\caption{A partial configuration}
\label{fig1}
\end{figure}
\end{center}

\begin{lem} \label{lem:partial}
For any non-zero $x,y,z$ belonging to $\C$ (respectively $\R)$, the Greechie diagram depicted in Figure~\ref{fig1} occurs as a sub-Greechie diagram of that of the orthogonality space $\mc{P}(\C^3)$ (respectively $\mc{P}(\R^3))$, where the elements of this diagram are given by the following one-dimensional subspaces
\end{lem}
\vspace{-5ex}

\begin{align*}
a_1&=\<1,0,0\> && b_1=\<0,y,z\> &&  c_1=\<0,\ol{z},-\ol{y}\> && d_1=\<-y\ol{y}-z\ol{z},\ol{x}y,\ol{x}z\> \\
a_2&=\<0,1,0\> && b_2=\<x,0,z\> &&  c_2=\<\ol{z},0,-\ol{x}\> && d_2=\<x\ol{y},-x\ol{x}-z\ol{z},\ol{y}z\>\\
a_3&=\<0,0,1\> && b_3=\<x,y,0\> &&  c_3=\<\ol{y},-\ol{x},0\> && d_3=\<x\ol{z},y\ol{z},-x\ol{x}-y\ol{y}\>
\end{align*}
\vspace{-8ex}

\begin{align*}
b_{12}&=\<\ol{yz},\ol{xz},-\ol{xy}\>&&
b_{112}=\<xy\ol{y}+xz\ol{z},-yz\ol{z},y\ol{y}z\> &&
b_{212}=\<-xz\ol{z},x\ol{x}y+yz\ol{z},x\ol{x}z\> \\
b_{13}&=\<\ol{yz},-\ol{xz},\ol{xy}\>&& 
b_{113}=\<xy\ol{y}+xz\ol{z},yz\ol{z},-y\ol{y}z\>
&&
b_{313}=\<-xy\ol{y},x\ol{x}y,x\ol{x}z+y\ol{y}z\>\\
b_{23}&=\<-\ol{yz},\ol{xz},\ol{xy}\>&& 
b_{223}=\<xz\ol{z},yz\ol{z}+x\ol{x}y,-x\ol{x}z\>&&
b_{323}=\<xy\ol{y},-x\ol{x}y,x\ol{x}z+y\ol{y}z\>
\end{align*}
\vspace{-3ex}

\hspace{2ex} $u=\<x,y,z\>$

\begin{proof}
The same proof holds for the real and complex cases. Since $x,y,z$ are non-zero, $u$ is not orthogonal to any $a_i$ for $i=1,2,3$ and this implies that for a given value of $i$, the elements $u,a_i,b_i,c_i,d_i$ are distinct and the orthogonalities among them are exactly those shown in Figure~\ref{fig1}. Further, the expressions for these elements in terms of $x,y,z$ are the ones given by cross products. If for $i\neq j$ some element of $\{a_i,b_i,c_i,d_i\}$ were equal to some element of $\{a_j,b_j,c_j,d_j\}$ then we would have a loop of order 4 or less, and this is impossible. So all the elements among the collection $a_i,b_i,c_i,d_i$ for $i=1,2,3$ are distinct and the orthogonalities among them are at least the ones shown in the figure. By looking at the expressions for these elements, it is clear that $a_i\not\perp b_j,c_j,d_j$ for $i\neq j$. By computing inner products and noting that none of $x,y,z$ are 0, we see also that $b_i\not\perp b_j,c_j$ for $i\neq j$ since these inner products are all a product of some of the variables $x,y,z$ and their conjugates. Note that $b_i$ cannot be orthogonal to $d_j$ for $i\neq j$ since this would produce a 4-loop. Thus, not only are the elements in the collection $a_i,b_i,c_i,d_i$ for $i=1,2,3$ all distinct, the orthogonalities among them are exactly the ones shown. We let $X=\{a_i,b_i,c_i,d_i\mid i=1,2,3\}$. 

For $i< j$ we have shown that $b_i\not\perp b_j$, so for given $i<j$ the elements $b_i,b_j,b_{ij},b_{iij},b_{jij}$ are distinct and the orthogonalities among them are exactly those shown in the figure. Further, the expressions for them in terms of $x,y,z$ are those given by cross products. For $i<j$, the element $b_{ij}$ does not belong to $X$ since $b_{ij}$ is orthogonal to both $b_i$ and $b_j$ and none of the elements in $X$ have this property. If $b_{iij}$ were to belong to $X$, it could only be one of $a_i$ or $c_i$ since $b_{iij}\perp b_i$ and $a_i,c_i$ are the only elements in $X$ orthogonal to $b_i$. But none of the entries in the description of $b_{iij}$ can be 0, so $b_{iij}$ does not belong to $X$, and a similar argument shows that $b_{jij}$ does not belong to $X$. Finally, none of the elements among the $b_i,b_{ij},b_{iij},b_{jij}$ can be equal to each other since this would produce a loop of order 4 or less from the hexagon in the diagram. Thus all of the elements in this diagram are distinct and all shown orthogonalities in the figure occur. So this figure is a sub-Greechie diagram of the Greechie diagram of the orthogonality space. 
\end{proof}

The proof of the previous lemma comes close to showing that Figure~\ref{fig1} is a full Greechie sub-diagram. The only possible additional orthogonalities that are not shown in the Figure~\ref{fig1} must involve an element of the hexagon. To have an additional orthogonality between two elements of the hexagon, the only possibility that does not produce a 4-loop is between two of the middle points on opposite sides, such as $b_{212}$ and $b_{313}$, but computing such inner products shows that they are always strictly positive real numbers, so this is not possible. So the only possible additional orthogonalities that are not shown in the Figure~\ref{fig1} must involve an element of the hexagon that is not one of the $b_i$ and an element that does not lie on the hexagon. This however can happen, and is the basis of the following definition. 

\begin{defn}
For a block $\<a_1\>,\<a_2\>,\<a_3\>$ of $\mc{P}(\C^3)$ (respectively $\mc{P}(\R^3)$), a one-dimensional subspace $\<u\>$ is a center for this block if the diagram in Figure~2 is a sub-Greechie diagram of the orthogonality space of $\mc{P}(\C^3)$ (respectively $\mc{P}(\R^3)$). \end{defn}

For vectors $u,v$ in $\C^3$ or $\R^3$, the angle $\theta$ between the subspaces $\<u\>$ and $\<v\>$ is given by 
\[
\cos\theta\,\,=\,\, \frac{|\<\,u,v\,\>|}{\|u\|\,\|v\|}.
\]

\noindent A one-dimensional subspace $\<u\>$ is unbiased with respect to a block $\<a_1\>,\<a_2\>,\<a_3\>$ if the angle between $\<u\>$ and $\<a_i\>$ is the same for each $i=1,2,3$. 
\vspace{-5ex}

\begin{center}
\begin{figure}[h]
\centering
\begin{tikzpicture}[scale=.65]
\draw[very thick] (0,0) circle(5);  
\fill (0:0) circle(.12);
\fill (90:5) circle(.12);
\fill (210:5) circle(.12);
\fill (-30:5) circle(.12);
\fill (90:3.5) circle(.07);
\fill (90:2) circle(.07);
\fill (210:3.5) circle(.07);
\fill (210:2) circle(.07);
\fill (-30:3.5) circle(.07);
\fill (-30:2) circle(.07);
\fill (135:1.41) circle(.07);
\fill (15:1.41) circle(.07);
\fill (255:1.41) circle(.07);
\draw (90:2) arc(90:270:1);
\draw (90:2)--(90:5);
\draw (-30:2) arc(-30:150:1);
\draw (-30:2)--(-30:5);
\draw (210:2) arc(210:390:1);
\draw (210:2)--(210:5);
\node at (-.4,-.3) {$u$};
\node at (85:5.3) {$a_1$};
\node at (82:3.7) {$b_1$};
\node at (78:2) {$c_1$};
\node at (205:5.4) {$a_2$};
\node at (202:3.7) {$b_2$};
\node at (198:2) {$c_2$};
\node at (-35:5.4) {$a_3$};
\node at (-38:3.7) {$b_3$};
\node at (-42:2) {$c_3$};
\fill (150:3.5) circle(.07);
\fill (180:3.03) circle(.07);
\fill (120:3.03) circle(.07);
\fill (60:3.03) circle(.07);
\fill (30:3.5) circle(.07);
\fill (0:3.03) circle(.07);
\fill (-60:3.03) circle(.07);
\fill (270:3.5) circle(.07);
\fill (240:3.03) circle(.07);
\draw (90:3.5)--(150:3.5)--(210:3.5)--(270:3.5)--(-30:3.5)--(30:3.5)--(90:3.5);
\node at (150:4.1) {$b_{12}$};
\node at (30:4.1) {$b_{13}$};
\node at (270:3.9) {$b_{23}$};
\node at (18:1.9) {$d_3$};
\node at (138:1.9) {$d_1$};
\node at (258:1.9) {$d_2$};
\node at (120:3.6) {$b_{112}$};
\node at (60:3.6) {$b_{113}$};
\node at (180:3.7) {$b_{212}$};
\node at (0:3.7) {$b_{313}$};
\node at (240:3.5) {$b_{223}$};
\node at (300:3.7) {$b_{323}$};
\node at (270:5.5) {$ $};
\draw[dashed] (150:3.5) to[out=-60,in=180] (-30:2);
\draw[dashed] (30:3.5) to[out=180,in=60] (210:2);
\draw[dashed] (270:3.5) to[out=60,in=-60] (90:2);
\fill (90:.9) circle(.07);
\fill (210:.9) circle(.07);
\fill (-30:.9) circle(.07);
\end{tikzpicture}   
\caption{Diagram for $\<u\>$ to be a center of the block of $\<a_1\>$, $\<a_2\>$, $\<a_3\>$.}
\label{fig2}
\end{figure}
\end{center}

\vspace{-5ex}

\begin{thm}
In either $\mc{P}(\C^3)$ or $\mc{P}(\R^3)$, a one-dimensional subspace $\<u\>$ is a center for a block $\<a_1\>,\<a_2\>,\<a_3\>$ iff $\<u\>$ is unbiased with respect to this block.    
\end{thm}

\begin{proof}
Since a unitary transformation preserves angles and there is a unitary transformation taking any block to any other, we may assume without loss of generality that $\<a_1\>=\<1,0,0\>$, $\<a_2\>=\<0,1,0\>$, $\<a_3\>=\<0,0,1\>$ and $\<u\>=\<x,y,z\>$. We may further assume that $\<u\>$ is not orthogonal to any of the $\<a_i\>$ for $i=1,2,3$ since this is the case if either $\<u\>$ is a center of this block or if $\<u\>$ is unbiased with respect to this block. Thus all of $x,y,z$ are non-zero and we can apply Lemma~\ref{lem:partial} to obtain that the configuration of Figure~\ref{fig1} is a sub-Greechie diagram and its elements are as described in terms of $x,y,z$ in that lemma. 

By definition, $\<u\>$ is a center for $\<a_1\>,\<a_2\>,\<a_3\>$ iff Figure~\ref{fig2} occurs as a sub-Greechie diagram, and in view of the fact that Figure~\ref{fig1} is a sub-diagram, this occurs iff $c_1\perp c_{23}$, $c_2\perp b_{13}$ and $c_3\perp b_{12}$. Note that the inner products are given by 
\[
\<\, c_1, b_{23}\,\>\,=\, xz\ol{z}-xy\ol{y}\quad\quad \<\,c_2,b_{13}\,\>\,=\, yz\ol{z}-yx\ol{x}\quad\quad \<\,c_3,b_{12}\,\>\,=\, zy\ol{y}-zx\ol{x}.
\]
Since $x,y,z$ are non-zero, having the indicated orthogonalities for $\<u\>$ to be a center is equivalent to having $|x|=|y|=|z|$, which in turn is equivalent to having $|\<\,a_i,u\,\>|$  for $i=1,2,3$ all be equal, hence to having $\<u\>$ be unbiased with respect to $\<a_1\>,\<a_2\>,\<a_3\>$.
\end{proof}

A collection of orthonormal bases of a finite-dimensional inner product space is said to be mutually orthogonal if the angle between $\<u\>$ and $\<v\>$ is the same whenever $u$ and $v$ come from different bases. It is well known that one can find 4 mutually unbiased bases in $\C^3$, for instance, using $\omega$ for a cube root of unity, the orthonormal bases are the appropriately scaled versions of  
\begin{center}
\begin{tabular}{ccccc}
$(1,0,0)$&&$(0,1,0)$&&$(0,0,1)$\\
$(1,1,1)$&&$(1,\omega,\omega^2)$&&$(1,\omega^2,\omega)$\\
$(1,\omega,\omega)$&&$(1,\omega^2,1)$&&$(1,1,\omega^2)$\\
$(1,\omega^2\omega^2)$&&$(1,\omega,1)$&&$(1,1,\omega)$
\end{tabular}
\end{center}

\noindent It is also well known that there is not even a pair of mutually unbiased bases of $\R^3$. In fact, there are not even two orthogonal vectors in $\R^3$ that are unbiased with respect to the same orthonormal basis. It is enough to show this for the standard basis, and it is clear,  that $u=\<x,y,z\>$ is unbiased with respect to the standard basis iff $|x|=|y|=|z|$ and it is not possible to have two orthogonal non-zero vectors in $\R^3$ satisfy this property since one cannot find an assignment of signs with $\pm 1\pm 1\pm 1=0$. 

\begin{thm}
There is a Greechie diagram that can be embedded into the Greechie diagram of the orthogonality space of $\mc{P}(\C^3)$ but not into that of $\mc{P}(\R^3)$.     
\end{thm}

\begin{proof}
Take two copies of the Greechie diagram of Figure~2 and glue them together by identifying the elements on the circular rims, then add a line connecting the two centers to indicate that they are orthogonal. If desired, add a third point to this line so that all blocks have 3 elements, but this is not necessary. The resulting configuration is a sub-Greechie diagram of the orthogonality space of $\mc{P}(\C^3)$ since it occurs with the elements of the circular rim given by the standard basis vectors $(1,0,0), (0,1,0), (0,0,1)$ and with the two centers being two elements from one of the mutually unbiased bases described above, for instance $(1,1,1)$ and $(1,\omega,\omega^2)$. This configuration does not occur as a sub-Greechie diagram of $\mc{P}(\R^3)$ since we noted that two centers of a block in $\R^3$ cannot be orthogonal.   
\end{proof}

\section{The second result} To begin this section, we first briefly describe some of the results of Tau and Tserunyan on embedding graphs into the orthogonality space of a Hilbert space. This is for the convenience of the reader, and also, as far as we know, this result is recorded only on the website \cite{Tau}. Our presentation of the result that we provide is modified from that given in \cite{Tau}.

\begin{thm} \label{thm:Tao}
Let $G$ be a finite graph whose vertices are $v_1,\ldots,v_n$. Then there is a basis $u_1,\ldots,u_n$ of $\mathbb{R}^n$ such that $v_i$ and $v_j$ are adjacent iff $u_i$ and $u_j$ are orthogonal. Furthermore, such a basis can be chosen so that the inner product of $u_i$ and $u_j$ is $\geq 0$ for all $i,j$. 
\end{thm}

\begin{proof}
We proceed by induction. The base case is trivial. Suppose $G$ has $n+1$ vertices and let $u_1,\ldots,u_n$ be a basis of $\mathbb{R}^n$ given by the inductive hypothesis for the subgraph $\{v_1,\ldots,v_n\}$. We will produce vectors $w_1,\ldots,w_{n+1}$ in $\mathbb{R}^n\oplus\mathbb{R}^n$ that span an $n+1$-dimensional subspace and are such that $w_i$ and $w_j$ are orthogonal iff $v_i$ and $v_j$ are adjacent and with the inner product of $w_i$ and $w_j$ positive for all $i,j$. This will prove the theorem. Let $e_1,\ldots,e_n$ be an orthonormal basis of $\mathbb{R}^n$. For $i\leq n$ let $w_i=u_i$ if $v_i$ is adjacent to $v_{n+1}$ and set $w_i=u_i+e_i$ otherwise. Set $w_{n+1}=e_1+\cdots+e_n$. Then for $i,j\leq n+1$  inner product of $w_i$ and $w_j$ is $\geq 0$ with equality iff $v_i$ is adjacent to $v_j$ and a routine argument shows that $w_1,\ldots,w_{n+1}$ are linearly independent. 
\end{proof}

Using this, one can show that every at most countable graph $G$ can be strongly embedded into the orthogonality graph of a separable Hilbert space $\mc{H}$ meaning that there is a one-one mapping $f:G\to \mc{H}$ such that vertices $u$ and $v$ are adjacent in $G$ iff $f(u)$ and $f(v)$ are orthogonal. Indeed, let the vertices of $G$ be enumerated $(v_n)_\mathbb{N}$ and for each $n$ let $u_{1,n},\ldots,u_{n,n}$ be a basis of $\mathbb{R}^n$ as given by the theorem for the subgraph $\{v_1,\ldots,v_n\}$ of $G$. Let $\mc{H}$ be the Hilbert space sum $\bigoplus_\mathbb{N}\mathbb{R}^n$ and for each $k\in\mathbb{N}$ set 
\[
u_k\,\,=\,\,\sum_{n=k}^\infty\, \frac{1}{2^n}u_{k,n}.
\]
\vspace{1ex}

\noindent Since inner products of the chosen basis vectors in each component $\mathbb{R}^n$ are positive it follows that $u_i$ and $u_j$ are orthogonal iff $v_i$ and $v_j$ are adjacent. 
\vspace{1ex}

The main result of Tau and Tserunyan is that there is an uncountable bipartite graph that cannot be strongly embedded into the orthogonality graph of any Hilbert space. We will not reproduce this in detail here since it is not as pertinent to our investigation, but only provide a brief a brief sketch. The key point is the following lemma. 

\begin{lem}[Tau]
If $(u_i)_I$ is an infinite family of unit vectors in a (perhaps non-separable) Hilbert space and the inner product $\langle u_i,u_j\rangle$ takes constant value $\theta$, then $\theta\geq 0$ and there is a vector $v$ and an orthonormal family $(w_i)_I$ all of which are orthogonal to $v$ and with $u_i=v+\sqrt{1-\theta}w_i$ for each $i\in I$.    
\end{lem}

The lemma is established by showing the convergence of the net $(v_F)$ indexed over the finite subsets $F\subseteq I$ where $v_F=\frac{1}{|F|}\sum \{v_i\mid i\in F\}$ is the ``average''. To construct the bipartite graph, one takes a set $\kappa$ of sufficiently large cardinality and constructs a bipartite graph between the sets of vertices $U=\{u_J\}$ where $J$ ranges over the subsets of $\kappa$ and $V=\{v_i\}$ where $i$ ranges over $\kappa$, and then putting an edge between $u_J$ and $v_i$ iff $i\in J$. If $U,V$ are unit vectors in a Hilbert space $\mc{H}$, since there only continuum possibilities for the value of an inner product, Erd\"os-Rano gives an uncountable subset $J\subseteq \kappa$ with all inner product $\langle v_i,v_j\rangle$ taking the same value $\theta$. Partitioning $J$ into two uncountable sets $J_0$ and $J_1$. Then $u_{J_0}$ is orthogonal to each $v_i$ for $i\in J_0$ so is orthogonal to $v$, and $u_{J_0}$ is non-orthogonal to each $v_i=v+\sqrt{1-\theta}w_i$ for $i\in J_1$, giving $u_{J_0}$ is non-orthogonal to all of the uncountably many orthonormal vectors $w_i$ for $i\in J_1$. 

Theorem~\ref{thm:Tao} shows that every finite graph $G$ can be strongly embedded into the orthogonality graph of a finite-height \ts{oml} namely a projection lattice $\mc{P}(\mathbb{R}^n)$ where $n$ is $|G|$. If we broaden our perspective beyond projection lattices we can go farther.

\begin{thm} \label{thm:part 2 main}
Every finite graph can be strongly embedded into the orthogonality graph of a finite \ts{oml}.
\end{thm}

To prove Theorem~\ref{thm:part 2 main} we first establish several lemmas. For these, we find it convenient to introduce some notation. For an \ts{oml} $L$ and $a\in L$ we write ${\uparrow}_L\,a$ for $\{b\in L\mid a\leq b\}$. The subscript is because we often consider such upsets in subalgebras of a given \ts{oml}. For elements $x,y$ of an \ts{oml} we write $x\perp y$ to indicate $x,y$ are orthogonal. For elements $u,v$ of a graph $G$ we write $u\perp v$ to indicate that $u,v$ are adjacent. 

\begin{lem} \label{lem:2}
For a finite \ts{oml} $L$ there is a finite \ts{oml} $M$ and a one-one map $g:L\to M$ such that  
\vspace{1ex}

\begin{enumerate}
\item $x\perp y$ $\Leftrightarrow$ $g(x)\perp g(y)$ for all $x,y\in L\setminus\{0\}$,
\item $g(x)\leq\bigvee g(A)$ $\Leftrightarrow$ $x\in A$ for all $x\in L$ and $A\subseteq L$.
\end{enumerate}
\end{lem}

\begin{proof}
Let $\mc{P}(L)$ be the powerset of $L$ and note that this is a Boolean algebra, hence an \ts{oml}. Set $M=L\times\mc{P}(L)$ and note that $M$ is a finite \ts{oml}. Define $g(x)=(x,\{x\})$. If $x,y\in L\setminus\{0\}$, then $x\perp y$ iff $g(x)\perp g(y)$ since distinct atoms $\{x\}$ and $\{y\}$ are orthogonal in any Boolean algebra. Note that $0\perp 0$, yet $g(0)\not\perp g(0)$, so the described qualification is necessary. If $x\in L$ and $A\subseteq L$, then $x\in A$ trivially implies $g(x)\leq\bigvee g(A)$ as each element of a set lies beneath its join. Conversely, if $g(x)=(x,\{x\})$ lies under $\bigvee g(A)=(\bigvee A,A)$, then $\{x\}\subseteq A$, so $x\in A$. 
\end{proof}

For the next step in the construction, we use a technique known as ``Kalmbach's coatom extension'' \cite[p.~310]{Kalmbach}. To describe this, let $L$ be an \ts{oml} and $e\in L\setminus\{0\}$. Then $[0,e']\cup[e,1]$ is a subalgebra of $L$. It is also a subalgebra when $e=0$, but this case would make pathologies in our treatment and we exclude it. Then $([0,e']\cup[e,1])\times 2$ is an \ts{oml} that has a subalgebra $([0,e']\times\{0\})\cup([e,1]\times\{1\})$ that is isomorphic to the subalgebra $[0,e']\cup[e,1]$ of $L$. 
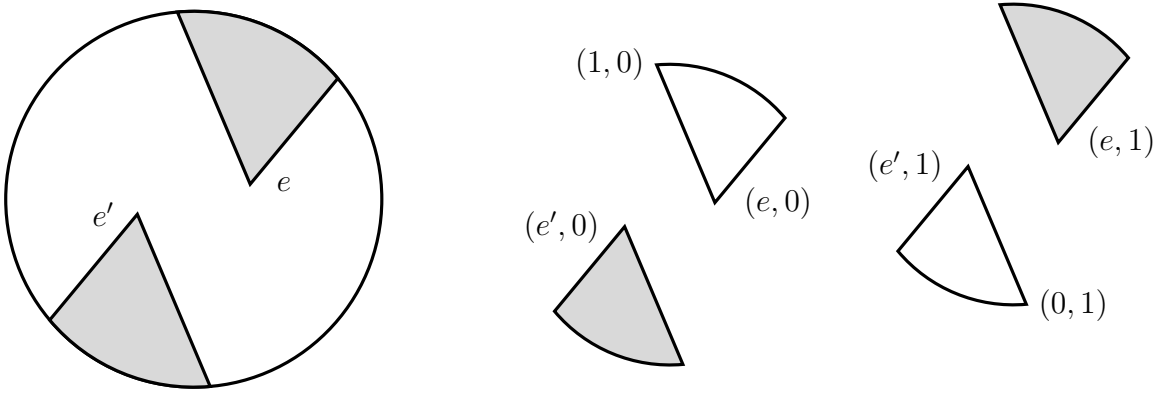
\vspace{2ex}

\begin{center}
\begin{tikzpicture}
\draw[very thick] (0,0) circle(2.5);
\draw[very thick,fill=gray!30] (-.75,-.2) --  (220:2.5) arc(220:275:2.5) -- cycle;
\draw[very thick,fill=gray!30] (.75,.2) --  (40:2.5) arc(40:95:2.5) -- cycle;
\node at (-1.2,-.2) {$e'$};
\node at (1.2,.2) {$e$};
\end{tikzpicture}
\hspace{8ex}
\begin{tikzpicture}[scale=.8]
\draw[very thick,fill=gray!30] (-.75,-.2) --  (220:2.5) arc(220:275:2.5) -- cycle;
\draw[very thick] (.75,.2) --  (40:2.5) arc(40:95:2.5) -- cycle;
\node at (-1.8,-.2) {$(e',0)$};
\node at (1.8,.2) {$(e,0)$};
\node at (0,-2.8) {$ $};
\node at (-1,2.5) {$(1,0)$};
\end{tikzpicture}
\hspace{1ex}
\begin{tikzpicture}[scale=.8]
\draw[very thick] (-.75,-.2) --  (220:2.5) arc(220:275:2.5) -- cycle;
\draw[very thick,fill=gray!30] (.75,.2) --  (40:2.5) arc(40:95:2.5) -- cycle;
\node at (-1.8,-.2) {$(e',1)$};
\node at (1.8,.2) {$(e,1)$};
\node at (0,-3.8) {$ $};
\node at (1,-2.5) {$(0,1)$};
\end{tikzpicture}
\end{center}
\vspace{2ex}

The \ts{oml}s $L$ and $([0,e']\cup[e,1])\times 2$ can be ``pasted'' over the isomorphic subalgebras that are shaded in the figure. The resulting ordering and orthocomplementation are effectively the union of those from the components. This is an instance of the technique known as Greechie's ``paste job'' and found in \cite[p.~306]{Kalmbach} and in more detail in \cite{Greechie}. Note that the element $a=(e,0)$ is an atom under $e\cong(e,1)$. Properties of the coatom extension are described below. 

\begin{lem} \label{lem:coatom}
Let $L$ be a finite \ts{oml} and $e\in L\setminus\{0\}$. Then there is a finite \ts{oml} $M$ with $L\leq M$ and an atom $a\in M$ with 
\vspace{1ex}
\begin{enumerate}
\item $a\not\in L$ and $a<e$,
\item ${\uparrow}_M\,a\cap L={\uparrow}_L\,e={\uparrow}_M\, e$,
\item If $x$ is an atom of $L$ and $x\neq e$, then $x$ is an atom of $M$.
\end{enumerate}
\end{lem}

\begin{lem} \label{lem: bigcoatom}
If $L$ is a finite \ts{oml} and $x_1,\ldots,x_n\in L\setminus\{0\}$ are distinct, then there is a finite \ts{oml} $M$ with $L\leq M$ and distinct atoms $a_1,\ldots,a_n\in M\setminus L$ with ${\uparrow}_M\,a_i\cap L={\uparrow}_L\,x_i$ and $a_i\perp a_j$ iff $x_i\perp x_j$ for all $i,j\leq n$.     
\end{lem}

\begin{proof}
The proof is by induction on $n$. When $n=1$, use Lemma~\ref{lem:coatom}. Now suppose $x_1,\ldots,x_{n+1}\in L\setminus\{0\}$ are distinct. By the inductive hypothesis there is a finite \ts{oml} $M$ with $L\leq M$ and distinct atoms $a_1,\ldots,a_n\in M\setminus L$ with ${\uparrow}_M\,a_i\cap L={\uparrow}_L\,x_i$ and $a_i\perp a_j$ iff $x_i\perp x_j$ for $i,j\leq n$. Apply Lemma~\ref{lem:coatom} to $M$ using $e=x_{n+1}$ to obtain a finite \ts{oml} $N$ with $M\leq N$ and an atom $a_{n+1}\in N\setminus M$ with ${\uparrow}_N\,a_{n+1}\cap M={\uparrow}_M\,x_{n+1}$, and so by intersecting both sides of this equality with $L$ and using that $x_{n+1}\in L$, with ${\uparrow}_N\,a_{n+1}\cap L={\uparrow}_L\,x_{n+1}$. So we have a finite \ts{oml} $N$ with $L\leq N$ and distinct atoms $a_1,\ldots,a_{n+1}\in N\setminus L$ with ${\uparrow}_N\,a_i\cap L={\uparrow}_L\,x_i$ for each $i\leq n+1$. It remains to show the statements about orthogonality. 

For $i,j\leq n+1$, if $x_i\perp x_j$, then since $a_i\leq x_i$ and $a_j\leq x_j$, we have $a_i\perp a_j$. It remains to show that if $i,j\leq n+1$ and $a_i\perp a_j$, then $x_i\perp x_j$. If $i,j\leq n$, this is given by our inductive hypothesis. Suppose for the remaining case that $i\leq n$ and $a_i\perp a_{n+1}$. Then $a_{n+1}\leq a_i'$. Since $a_i'\in M$ and ${\uparrow}_N\,a_{n+1}\cap M={\uparrow}_M\,x_{n+1}$, we have $x_{n+1}\leq a_i'$. Then $a_i\leq x_{n+1}'$, and since $x_{n+1}'\in L$ and ${\uparrow}_N\,a_i\cap L={\uparrow}_L\, x_i$ we have $x_i\leq x_{n+1}'$, hence $x_i\perp x_{n+1}$.
\end{proof}

\begin{lem} \label{lem: stu}
If $G$ is a finite loopless undirected graph, then there is a finite \ts{oml} $L$ and a one-one map $f:G\to L\setminus\{0\}$ with $x\perp y$ iff $f(x)\perp f(y)$ for all $x,y\in G$.      
\end{lem}

\begin{proof}
By induction on $|G|=n$. For $n=1$ we can map the single vertex of $G$ to the top element of the 2-element Boolean algebra. Suppose $|G|=n+1$. If $G$ is a complete graph, then we can map $G$ to the atoms of a finite Boolean algebra with $n+1$ atoms. So assume there is a vertex $w\in G$ that is not adjacent to some other vertex. Let $H=G\setminus\{w\}$ and let $P$ be the set of vertices adjacent to $w$, all of which belong to $H$, and $Q=H\setminus P$. 
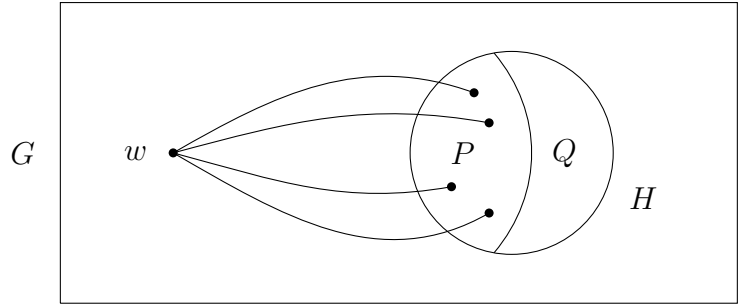
\vspace{2ex}

\begin{center}
\begin{tikzpicture}
\draw (0,0) circle(1.35);    
\draw (100:1.35) to[out=-50,in=50] (260:1.35);
\fill (-.5,.8) circle(.06);
\fill (-.3,.4) circle(.06);
\fill (-.3,-.8) circle(.06);
\fill (-.8,-.45) circle(.06);
\node at (-.65,0) {$P$};
\node at (.7,0) {$Q$};
\node at (-5,0) {$w$};
\fill (-4.5,0) circle(.06);
\draw (-4.5,0) to[out=30,in=160] (-.5,.8);
\draw (-4.5,0) to[out=15,in=170] (-.3,.4);
\draw (-4.5,0) to[out=-15,in=190] (-.8,-.45);
\draw (-4.5,0) to[out=-30,in=210] (-.3,-.8);
\node at (1.75,-.6) {$H$};
\draw (-6,-2)--(-6,2)--(3,2)--(3,-2)--(-6,-2);
\node at (-6.5,0) {$G$};
\end{tikzpicture}    
\end{center}
\vspace{2ex}

Since the subgraph $H$ of $G$ has $n$ elements, the inductive hypothesis gives a finite \ts{oml} $L$ and a one-one mapping $h:H\to L\setminus\{0\}$ with $x\perp y$ iff $h(x)\perp h(y)$. By Lemma~\ref{lem:2} there is a finite \ts{oml} $M$ and a one-one map $g:L\to M$ such that $u\perp v$ iff $g(u)\perp g(v)$ for all $u,v\in L\setminus\{0\}$ and with $g(u)\leq\bigvee g(A)$ iff $u\in A$ for all $u\in A$ and $A\subseteq L$. Then the composite $k=g\circ h$ satisfies $x\perp y$ iff $k(x)\perp k(y)$ for all $x,y\in H$ and $k(x)\leq \bigvee k(A)$ iff $x\in A$ for all $x\in H$ and $A\subseteq H$. 

Since we have assumed that $w$ is not adjacent to some vertex in $G$, we have that $P\neq H$. Let $e\in M$ be such $e'=\bigvee k(P)$. For all $x\in H$ we have $k(x)\leq\bigvee k(P)$ iff $x\in P$, therefore $e'\neq 1$, and hence $e\neq 0$. Apply the coatom extension Lemma~\ref{lem:coatom} to $M$ with the element $e\in M\setminus\{0\}$ to obtain a finite \ts{oml} $N$ with $M\leq N$ and $a\in N\setminus M$ such that $a$ is an atom of $M$ and ${\uparrow}_M\, a\cap N={\uparrow}_M\,e$. Thus for $x\in H$ we have $x\in P$ iff $k(x)\leq e'$ iff $e\leq k(x)'$, and since $k(x)'\in M$, this occurs iff $a\leq k(x)'$, hence iff $k(x)\perp a$. Define $f:G\to N$ by 
\[
f(x)\,\,=\,\,\begin{cases}
    \,k(x)&\mbox{if $x\in H$}\\
    \,\,\,\,a&\mbox{if $x=w$}
\end{cases}
\]
Since $a\not\in M$ and $k$ is one-one, $f$ is one-one. Since $M\leq N$, if $x,y\in H$ and $x\perp y$ iff $k(x)\perp k(y)$, for $x,y\in H$ we have $x\perp y$ iff $f(x)\perp f(y)$. Finally, for $x\in H$, by the definition of $P$ we have $x\perp w$ iff $x\in P$, hence iff $k(x)\leq e'$, and thus iff $f(x)=k(x)\perp a=f(w)$. 
\end{proof}

We return to the proof of Theorem~\ref{thm:part 2 main}. Suppose $G$ is a finite graph with vertices $x_1,\ldots,x_n$. By Lemma~\ref{lem: stu} there is a finite \ts{oml} $L$ and a one-one mapping $f:G\to L\setminus\{0\}$ such that $x_i\perp x_j$ iff $f(x_i)\perp f(x_j)$ for all $i,j\leq n$. By Lemma~\ref{lem: bigcoatom} there is a finite \ts{oml} $M$ and distinct atoms $a_1,\ldots,a_n\in M\setminus L$ with $a_i\perp a_j$ iff $f(x_i)\perp f(x_j)$ for all $i,j\leq n$. Let $h:G\to M\setminus\{0\}$ be given by setting $h(x_i)$ to be the atom $a_i$ of $M$. Then $f$ is one-one and $x_i\perp x_j$ iff $h(x_i)\perp h(x_j)$. 

\begin{cor}
Every graph can be strongly embedded into the orthogonality graph of an atomic \ts{oml}.    
\end{cor}

\begin{proof}
We use the fact that every finite graph can be strongly embedded into the orthogonality graph of an \ts{oml}. This is provided either by the result of Tao and Tserunyan or by Theorem~\ref{thm:part 2 main}. Suppose $G$ is a graph. Construct a first-order language having a constant symbol $c_v$ for each vertex $v$ of $G$, a binary predicate $E$ for edges, constants $0,1$, a unary operation symbol $'$ and binary operation symbols $\wedge$ and $\vee$. We consider a set $\Sigma$ of first-order sentences that include sentences for the \ts{oml} axioms and a sentence saying that each non-zero element has an atom beneath it (note that the order $\leq$ of a lattice can be expressed in a first-order way from its meet operation). For each vertex $v$ we include the sentence saying that $c_v$ is an atom, and for each pair of vertices $u,v$ of $G$ we include the sentence $E(c_u,c_v)$ if $u,v$ are adjacent in $G$ and we include the sentence $\neg E(c_u,c_v)$ if $u,v$ are not adjacent in $G$. Since every finite graph can be strongly embedded into the orthogonality graph of an atomic \ts{oml}, every finite subset of $\Sigma$ has a model. So by the compactness theorem, $\Sigma$ has a model, hence $G$ can be strongly embedded into the orthogonality graph of an atomic $\ts{oml}$.     
\end{proof}

\end{document}